\documentclass[11pt]{llncs}
\usepackage{amssymb,amsmath,color,tikz,subfig}

\usepackage[margin=1.15in]{geometry} 

\usepackage{graphicx} 
\usepackage[T1]{fontenc}
\usepackage[english]{babel}
\usepackage{eurosym}
\usepackage{palatino}
\usepackage{tikz}
\usepackage{verbatim}
\usepackage{fancybox}
 \usepackage{url}

\usepackage{subfig}
\usepackage{color}
\usepackage{enumerate}
\usepackage{epsfig}
\usepackage[active]{srcltx}
\usepackage{amsfonts}
\usepackage{amsmath}
\usepackage[mathlines]{lineno}
\usepackage{xspace}
\usepackage{tikz}
\usepackage{authblk}

\newtheorem{observation}[theorem]{Observation}

\newcommand{\restate}[3]{
\smallskip
\noindent {\bf #1~#2.}
\textit{#3}
\smallskip
}
\newcommand{\notshow}[1]{{}}

\newcommand{\restatethm}[2]{\restate{Theorem}{#1}{#2}}

\title{Welfare and Rationality Guarantees for the\\ Simultaneous Multiple-Round Ascending Auction}
\titlerunning{Welfare and Rationality Guarantees for the SMRA}
\toctitle{Welfare and Rationality Guarantees for the Simultaneous Multiple-Round Ascending Auction}

 \author{Nicolas Bousquet\inst{1} \and Yang Cai\inst{2}\and Adrian Vetta\inst{3}}
 \institute{
 Department of Mathematics and Statistics, McGill University, and LIRIS, Ecole Centrale Lyon\\
 \email{nicolas.bousquet@ec-lyon.fr}
 \and 
 School of Computer Science, McGill University\\
  \email{cai@cs.mcgill.ca}
 \and
 Department of Mathematics and Statistics, and School of Computer Science, McGill University\\
\email{vetta@math.mcgill.ca}
}

\begin{document}

\maketitle

\vspace{-10pt}

\begin{abstract}
The simultaneous multiple-round auction (SMRA) and the combinatorial clock auction (CCA) are
 the two primary mechanisms used to sell bandwidth. Under truthful bidding, the SMRA is known to 
output a Walrasian equilibrium that maximizes social welfare provided the bidder valuation functions
 satisfy the gross substitutes property~\cite{Mil00}. 
Recently, it was shown that the combinatorial clock auction (CCA) provides good welfare guarantees
for general classes of valuation functions~\cite{BCH15}. This motivates the question of whether
similar welfare guarantees hold for the SMRA in the case of general valuation functions.

We show the answer is no.
But we prove that good welfare guarantees still arise if the degree of complementarities in the 
bidder valuations are bounded. In particular, if bidder valuations functions are 
$\alpha$-near-submodular then, under \emph{truthful bidding}, the SMRA has a welfare ratio (the worst case ratio
between the social welfare of the optimal allocation and the auction allocation)
of at most $(1+\alpha)$. The special case of submodular valuations, namely $\alpha=1$, was studied in~\cite{FKL12}
and produces individually rational solutions. However, for $\alpha>1$, this is a bicriteria guarantee, to obtain good welfare under truthful bidding requires relaxing individual rationality. In particular,
it necessitates a factor $\alpha$ loss in the degree of individual rationality provided by the auction.
We prove this bicriteria guarantee is asymptotically (almost) tight.

Truthful bidding, though, is not reasonable assumption in the SMRA~\cite{Cra13}. But, bicriteria
guarantees continue to hold for natural bidding strategies that are \emph{locally optimal}. Specifically,
the welfare ratio is then at most $(1+\alpha^2)$ and the individual rationality guarantee is again at most $\alpha$, 
for $\alpha$-near submodular valuation functions. These bicriteria guarantees are also (almost) tight.

Finally, we examine what strategies are required to ensure individual rationality in the SMRA with general valuation
functions. First, we provide a weak characterization, namely \emph{secure bidding}, for individual rationality.
We then show that if the bidders use a profit-maximizing secure bidding strategy 
the welfare ratio is at most $1+\alpha$. 
Consequently, by bidding securely, it is possible to obtain the same welfare guarantees as truthful bidding
without the loss of individual rationality.
Unfortunately, we explain why secure bidding may be incompatible with the auxiliary
bidding activity rules that are typically added to the SMRA to reduce gaming.

\vspace{5pt}

\noindent{\bf Keywords:} ascending auctions, SMRA, welfare guarantee, individual rationality, near-submodular.
\end{abstract}

\vspace{-15pt}

\section{Introduction}
The question of how best to allocate spectrum dates back over a century, and
the case in favour of selling bandwidth was first formalized in the academic literature 
as far back as 1959 by Ronald Coase~\cite{Coa59}. Over the past twenty years there have been
large number spectrum auctions world-wide and, amongst these,
the Simultaneous Multi-Round Auction (SMRA) and the Combinatorial Clock Auction (CCA)
have proved to be extremely successful.

Both of these multiple-item auctions are based upon the same underlying mechanism. 
At time $t$, each item $j$ has a price $p^t_j$. Given the current prices, each bidder $i$
then selects her preferred set $S^t_i$ of items. The price of any item that has excess demand then
rises in the next time period and the process is repeated. 
There are important differences between the two auctions however.
The SMRA uses {\em item bidding}, that is, the auctioneer views the selection 
of $S_i^t$ as a collection of bids, one bid for every item of $S_i^t$. 
It also utilizes the concept  of a {\em standing high bid} \cite{CK81}.
Any item (with a positive price) has a {\em provisional winner}. That bidder will win the item unless a higher
bid is received in a later round. If such a bid is received then the standing high bid is increased
and a new provisional winner assigned (chosen at random in the case of a tie).
Item bidding and standing high bids lead to a major drawback, the \emph{exposure problem}. 
Namely, a large set may be desired but such a bid may result in being allocated only a smaller undesirable subset.
If the bidder valuation functions satisfy the gross-substitutes property 
then this problem does not arise. Indeed, given truthful bidding, Milgrom~\cite{Mil00} showed that 
the SMRA will terminate in a Walrasian Equilibrium that maximizes social welfare; see also \cite{KC82,GS99} who studied
a similar auction mechanism. The exposure problem is also absent when the bidder valuation functions 
are submodular. In that case, Fu et al.~\cite{FKL12} show that the final allocation, whilst not necessarily a
Walrasiam Equilibrium, does provide at least half of the optimal social welfare.\footnote{Their proof is not for the SMRA, but it can be adapted to apply there.} For these classes of valuation function, the SMRA is {\em individual rational}
in every time period. That is, if bidder $i$ is provisionally allocated set $S$ at round $t$ then the value of $S$ to $i$ is 
at least the price of $S$.

For broader classes of valuation function that permit {\em complementarities}, though,
the exposure problem does arise under the SMRA.
This is a practical issue because in spectrum auctions bidder valuation functions 
typically do exhibit complementarities. The CCA~\cite{PRR03} was designed to deal with such complementarities.
Specifically, the CCA uses {\em package bidding} rather than item bidding. 
A package bid is an all-or-nothing bid. Consequently, a bidder cannot be allocated a subset of her bid; in particular,
a bidder cannot be allocated an undesirable subset. 
Unfortunately, the basic CCA mechanism cannot provide for non-trivial approximate welfare guarantees,
even for auctions with additive valuation functions and a small number of bidders and items~\cite{BCH15}.
It is perhaps surprising, then, that a minor adjustment to the CCA mechanism leads to good welfare guarantees 
for {\em any} class of  valuation function. 
Specifically, if bid increments are made proportional to excess demand the welfare of the
CCA is within an $O(k^2\cdot \log n \log^2 m)$ factor of the optimal welfare~\cite{BCH15}.
Here $n$ is the number of items, $m$ is the number of bidders and $k$ is the maximum 
cardinality of a set desired by the bidders. The fact that the CCA
can generate high welfare for general valuation functions motivates the work in this paper.
Is it possible that the SMRA also performs well with general valuation functions?

\subsection{Our Results}
The short answer to the question posed above is {\sc no}, the SMRA
cannot guarantee high social welfare for valuation functions that exhibit
complementarities (see Section~\ref{sec:bad-smra}).

It turns out however that we can quantify precisely the welfare guarantee 
in terms of the magnitude of the complementarities exhibited by the valuation function.
To explain this we require a few definitions. 
Each bidder $i \in B$ has value $v_{i}(S)$ for any set of items $S\subseteq \Omega$. 
The valuation function $v_i()$ is monotonically non-decreasing (free-disposal).
Each bidder has a quasi-linear utility, that is, 
its \emph{utility} for a set $S$ is $v_{i}(S)-p(S)$, where $p(S)$ is the price of $S$. 
The social welfare of an allocation $\mathcal{S}=\{S_{1},\ldots, S_{n}\}$, 
where the $S_{i}$ are pairwise-disjoint subsets of the items, 
is $\omega(\mathcal{S}) = \sum_{i}v_{i}(S_{i})$.
Next, to quantify the extent of complementarities, let the \emph{degree of submodularity} \cite{AV14} of a function $f$ be
\[ \mathcal{D}(f) = \min_{x \in I} \min_{A,B : A \subset B} \frac{f(A \cup x) - f(A)}{f(B \cup x) - f(B)} \]
Note that $f$ is submodular if and only if $\mathcal{D}(f) \geq 1$. 
We say that $f$ 
is \emph{$\alpha$ near-submodular} if $\mathcal{D}(f) \geq \frac{1}{\alpha}$.
A similar concept to near-submodularity, called {\em bounded complementarity}, is introduced 
by Lehman et al. \cite{LLN06}. 

The parameter $\alpha$ turns out to be key in
explaining the performance of the SMRA. To explain this we require
one more concept. 
We say that a bidder $i \in B$ is $\lambda$-{\em individually rational} 
if $\lambda \cdot v_i(S_i^t) \geq p(S_i^t)$ in each round $t$.
Note that if $\lambda=1$ then we have individual rationality.
We say that an auction mechanism is $\lambda$-{\em individually rational} if every bidder
is $\lambda$-individually rational.
We then prove in Section \ref{sec:truthful}:
 
\begin{theorem}\label{thm:constantfactor}
If bidders have $\alpha$ near-submodular valuations then, under (conditional) truthful bidding~\footnote{A detailed discussion on truthful and conditional truthful biddings will follow in Section~\ref{sec:truthful-smra}.},\\
(i) The SMRA outputs an allocation $\mathcal{S}$ with $\omega(S)\ge \frac{1}{1+\alpha}\cdot \omega(\mathcal{S}^*)$ where
$\mathcal{S}^*$ is the optimal allocation.\\
(ii) The auction is $\alpha$-individually rational.
\end{theorem}
The bi-criteria guarantees in Theorem~\ref{thm:constantfactor} are (almost) tight. There are examples
with $\alpha$ near-submodular valuations where the SMRA is only $\alpha$-individually rational
and the welfare guarantee tends to $\frac{1}{1+\alpha}$; see Section~\ref{sec:tight-truthful}.
Despite the fact that SMRA has arbitrarily poor welfare guarantees, it seems to perform
very well in practice. Theorem~\ref{thm:constantfactor} provides an explanation for this, and confirms
empirical results, since complementarities exist but are typically bounded in magnitude in most spectrum auctions. 
Indeed, the SMRA has been proposed for auctions where valuation functions 
have weak complementarities~\cite{AC11}.

There are, however, two major drawbacks inherent in Theorem~\ref{thm:constantfactor}.
The first drawback is that it relies upon {\em truthful bidding}, that is, in each round the bidder selects the feasible set that
maximizes utility. But, as we explain in Section \ref{sec:truthful-smra}, there are many reasons why a
bidder will not bid truthfully in the SMRA.
One of these reasons is that, in a spectrum auction, a bidder may not even know its own valuation function~\cite{Cra13}.
Bidders typically can however make comparisons between similar sets. Thus, a natural method by which a bidder can
select a bid is via local improvement. 

We show, in Section~\ref{sec:locallyopt}, that local improvement leads to similar guarantees as truthful bidding 
(albeit with an additional $\alpha$ factor in
the denominator for the welfare guarantee).
\begin{theorem}\label{thm:localimprov}
If bidders have $\alpha$ near-submodular valuations then, under (conditional) local improvement bidding,\\
(i) The SMRA outputs an allocation $\mathcal{S}$ with 
$\omega(S)\ge \frac{1}{1+\alpha^2}\cdot \omega(\mathcal{S}^*)$ where
$\mathcal{S}^*$ is the optimal allocation.\\
(ii) The auction is $\alpha$-individually rational.
\end{theorem}
 Again, the bounds in Theorem~\ref{thm:localimprov} are (almost) tight.
 
The second drawback is that Theorem~\ref{thm:constantfactor} shows that the 
SMRA is not individually rational.
That is, it may produce outcomes which give negative utility to some bidders.
Consequently, in Section~\ref{sec:ir} we provide a detailed study of what bidding strategies
are required to ensure the individual rationality of the SMRA, and what are the consequences
for welfare when such strategies are used. Towards this end,
we characterize the individual rationality of the SMRA in terms of {\em secure bidding}.
We then prove, in Section~\ref{sec:welfare-secure}, that secure bidding has a 
good welfare guarantees, provided the bidders make profit maximizing secure bids.
\begin{theorem}\label{thm:profitmax}
If bidders have $\alpha$ near-submodular valuations then, under (conditional) profit maximizing secure bidding,
the SMRA outputs an allocation $\mathcal{S}$ with 
$\omega(S)\ge \frac{1}{1+\alpha}\cdot \omega(\mathcal{S}^*)$ where
$\mathcal{S}^*$ is the optimal allocation.
\end{theorem}
Consequently, by bidding securely, it is possible to obtain the same welfare guarantees as truthful bidding
without the loss of individual rationality!




\section{The Simultaneous Multiple-Round Ascending Auction}\label{sec:SMRA}

The SMRA was first proposed by Milgrom, Wilson and McAfee for the 1994 FCC spectrum auction.
It is an ascending price auction that simultaneously sells many items.
Let $B$ be a set of $n$ bidders and let $\Omega$ be a collection of $m$ items.
For each item $j\in \Omega$ the auction posits
an item-price $p_j^t$ at the start of round $t$. 
Moreover, the SMRA has a unique {\em standing high bidder} for each item with a positive price.
Specifically, at the start of round $t$, bidder $i$ is the standing high bidder for a set of items 
$S_i^t$; we call $S_i^t$ the {\em provisional (winning) set} for bidder~$i$. \vspace{5pt}


\noindent\underline{The SMRA mechanism}:
Initially $p^0_j=0$ for each item $j\in \Omega$, and $S_i^0=\emptyset$ for each bidder $i\in B$ and $t=0$.
The auction then iterates over rounds as follows.  
In round $t$, bidder $i$ bids for a set $T_i^t\subseteq \Omega \setminus S^t_i$ under the assumption 
that the price of each item $j\in \Omega\setminus S^t_i$ is incremented to $p_j^t+\epsilon$. 
We call $T^t_i$ the {\em conditional bid} for $i$. The term conditional is used as the 
auction mechanism automatically assumes that bidder $i$ also makes a bid of price $p_j^t$ for every 
item $j\in S^t_i$ (recall, bidder $i$ is the provisional winner of the items $S^t_i$).

The item-prices and provisional sets are then updated. Take an item $j$ and suppose that
$j$ is in exactly $k$ of the conditional bids. If $k=0$ then no bidder has
placed a bid on item $j$ at the incremented price $p_j^t+\epsilon$. Thus we set
$p_j^{t+1}=p_j^t$ and the standing high bidder for $j$ remains the same, $i.e.$
if $j\in S^t_i$ then $j\in S^{t+1}_i$. On the other hand if $k>0$ (we say that
$j$ is in \emph{excess demand}) then at least
one bidder accepted the incremented price $p_j^t+\epsilon$. Thus we set
$p_j^{t+1}=p_j^t+\epsilon$. The mechanism then randomly selects a bidder $i$ 
amongst these $k$ bidders and places $j\in S^{t+1}_i$. Note that, 
in this case, the standing high bidder must change as the previous
standing high bidder was only assumed to bid the non-increment price $p_j^t$.

The mechanism then proceeds to the next round. The auction terminates 
when the conditional bids $T^t_i$ of all the bidders are empty, at which point
each bidder $i$ is permanently allocated her provisional set $S^t_i$ for
a price $\sum_{j\in S^t_i} p^t_i$. \vspace{5pt}

An extremely important property of the SMRA is that the use of
standing high bidders implies that every item with a positive price is
sold.
\begin{observation}
In an SMRA auction, every item with a positive price is sold. \qed
\end{observation}

\subsection{Truthful Bidding in the SMRA}\label{sec:truthful-smra}
A key factor in determining the practical success of the auction is accurate
price discovery (see, for example Cramton \cite{Cra06,Cra13}).
This, in turn, relies upon bidding that is truthful or, at least,
approximately truthful. There are two pertinent issues here. 
Firstly, is the SMRA mechanism compatible with truthful bidding? Specifically,
the use of conditional bidding implicitly implies that bidders are
forced to rebid on their provisional sets.
However, suppose that $T^t_i$ is the optimal conditional bid, that is
$$T^t_i =\mathrm{argmax}_{T\subseteq \Omega\setminus S^t_i} \ 
\left(v_i(T\cup S^t_i) - v_i(S^t_i) - \sum_{j\in T} (p^t_j+\epsilon)\right)$$
It need not be the case that  the implicit bid $S^t_i\cup T^t_i$ is truthful. In particular, we may have
$$S^t_i\cup T^t_i  \neq \mathrm{argmax}_{T\subseteq \Omega} \ 
\left(v_i(T) - \sum_{j\in T\cap S^t_i}  p^t_j - \sum_{j\in T\setminus S^t_i} (p^t_j+\epsilon)\right)$$
Recall, here, that bidder $i$ has a personalized set of prices: 
$\left( ({\bf p})_{S^t_i}, ({\bf p}+\epsilon \cdot {\bf 1})_{\Omega\setminus S^t_i}\right)$.
Indeed, at round $t$, bidder $i$ has an $\epsilon$ discount on the prices of $S_i^t$.

Interestingly, truthful bidding is compatible with the SMRA (for any price trajectories) precisely if the
valuation function satisfies the gross substitutes property \cite{Mil00}. The gross substitutes property\footnote{A valuation function satisfies the 
{\em gross substitutes property} if, given any set of prices, increasing the price of some goods does not decrease demand
for another good.}
 was defined by Kelso and Crawford \cite{KC82} and used by them to prove 
the existence of Walrasian equilibrium. Moreover, with gross substitutes, 
the SMRA will converge to a Walrasian equilibrium; furthermore such an equilibrium will maximize
social welfare (given negligible price increments) -- see Milgrom \cite{Mil00,Mil04}.

Secondly, even if truthful bidding is compatible with the SMRA, it is unlikely that the bidders 
will actually bid truthfully. 
For example, in bandwidth auctions, firms typically have ranked bandwidth targets and budget constraints that
are more important than profit-maximization.
Moreover, the valuation function is often not known in advance, rather it is ``learned" as the auction proceeds.
Regardless, the SMRA and the CCA do both incorporate a set of bidding activity rules to encourage truthful bidding.
In the CCA these include revealed preference bidding rules that are difficult to game~\cite{ACM06,BV15}.
However, the bidding rules in the SMRA are weaker and strategic bidding is common -- examples
include demand reduction, parking, and hold-up strategies~\cite{Cra13}.

Consequently, as well as examining truthful (optimal conditional) bidding, we will
examine the natural strategy of local improvement bidding that consists of 
attempting to add one item, delete one item, or replace one item in the current proposed solution.
 Gul and Stacchetti \cite{GS99} prove that this local improvement method finds an optimal demand set, 
 given any set of prices, if the valuation function has the gross substitutes property.
 We examine the quality of outcomes, for more general valuation functions, when this local search 
 method is used in the SMRA in Section~\ref{sec:locallyopt}.
 
From now on, we concentrate on conditional bidding. Thus we will omit the term conditional when we refer to conditional
truthful bidding or conditional profit maximizing bidding.

\subsection{A Bad Example}\label{sec:bad-smra}
Unfortunately, the welfare ratio of the SMRA can be arbitrarily bad if the valuations exhibit complementarities.
This is the case even for auctions with just two bidders $\{1,2\}$ and two items $\{a,b \}$. 
Suppose both bidders have value $1$ for each individual, but value the pair of items at $M$, for some large value $M$.
Clearly, the optimal allocation has welfare $M$ and consists in allocating both items to one of the bidders.
However, the allocation of the SMRA has welfare $2$ with probability $\frac 12$. 
Indeed, the provisional set at round $t+1$ is the complement
of the provisional set at round $t$ since the conditional bid 
of each bidder will be the complement of her provisional set.
So the final allocation (which occurs when both prices exceed $\frac{M}{2}$) just depends on the allocation 
at the end of the first round; this allocates one item to each bidder with probability $\frac 12$ since it is randomized.
By further increasing the number of identical bidders, it can be shown that the probability of the low welfare outcome
tends to one. 

This simple example has an important implication.
Note that, at any round, both bidders bid on both items until the end of the auction. 
Thus, excess demand is constant in each round.
So, even if price increments depend on the excess demand, one cannot achieve a better welfare ratio . 
This contrasts sharply with the behavior of the CCA whose welfare ratio becomes
polylogarithmic if the price increment is allowed to depend upon the excess demand
and the size of the demand sets are bounded~\cite{BCH15}.

\section{Bicriteria Guarantees for the SMRA under Truthful Bidding}\label{sec:truthful}
We now prove Theorem~\ref{thm:constantfactor} and show
that, under truthful bidding, the worst case welfare and rationality guarantees
are dependent upon the degree of submodularity in the bidder valuation functions.
First, in Section~\ref{sec:rational-truthful} we prove the individual rationality guarantee.
Then, in Section~\ref{sec:welfare-truthful} we prove the welfare guarantee. Finally,
in Section~\ref{sec:tight-truthful} we show that these guarantees are (almost) tight.

\subsection{A Rationality Guarantee}\label{sec:rational-truthful}
\begin{theorem}\label{thm:smra-truthful}
Given $\alpha$-near-submodular truthful bidders, the SMRA outputs an $\alpha$-individually rational
allocation.
\end{theorem}
\begin{proof}
In order to show $\alpha$-individual rationality upon termination, let us prove a stronger 
result. Specifically, we will show that for any time $t$ and any bidder $i$, every set $S' \subseteq S_i^t$ 
satisfies $\alpha \cdot v_i(S') \geq p(S')$. 
We proceed by induction on $t$. The statement trivially holds for $t=0$. 
For the induction hypothesis, assume that bidder $i$ is allocated the set $S_i^t$ in round $t$ where 
\begin{equation}\label{eq:induction}
\alpha \cdot  v_i(S') \geq p^t(S') \ \ \ \ \ \forall S' \subseteq S_i^t
\end{equation} 
We now require the following claim:

%
 \begin{claim}\label{cl:near-rational1}
  Let $X \subseteq S_i^t \cup T_i^t$ be such that $\alpha \cdot v_i(X) \geq p^t(X\cap S_i^t)+ p^{t+1}(X \setminus S_i^t)$. 
  Then, for every $x\in T_i^t \setminus X$, we have
  \begin{equation*}
  \alpha \cdot v_i(X \cup x) \geq p^t(X\cap S_i^t)+ p^{t+1}(X \cup x \setminus S_i^t)
  \end{equation*}
\end{claim}  
  \begin{proof}
Take any $x\in T_i^t \setminus X$. By $\alpha$ near-submodularity, we have
 \begin{equation*}
\frac{v_i(X \cup x) - v_i(X)}{v_i(S_i^t\cup T_i^t) - v_i(S_i^t \cup T_i^t\setminus x)} \geq \frac{1}{\alpha}
 \end{equation*}
 Consequently,
  \begin{eqnarray} \label{eq:useful}
\alpha \cdot v_i(X \cup x) - \alpha \cdot v_i(X) 
&\geq& v_i(S_i^t\cup T_i^t) - v_i(S_i^t \cup T_i^t\setminus x)\nonumber\\
& \geq& p^{t+1}(x)\\
&=& p^t(x)+\epsilon \nonumber
  \end{eqnarray} 
Here the second inequality follows from truthful bidding. Otherwise,
$T_i^t \setminus x$ is a more profitable bid than $T_i^t$.
 The equality arises as $x\notin S_i$.
 
 By the condition in the statement of the claim, we have $ \alpha \cdot v_i(X) \geq p^t(X\cap S_i^t)+ p^{t+1}(X \setminus S_i^t)$.
Therefore
 \begin{eqnarray*} 
 \alpha \cdot v_i(X \cup x) 
 &\geq& p^t(X\cap S_i^t)+ p^{t+1}(X \setminus S_i^t) + p^{t+1}(x) \\
 &=& p^t(X\cap S_i^t)+ p^{t+1}(X\cup x \setminus S_i^t)
  \end{eqnarray*} 
Again, the equality arises as $x\notin S_i$.
  \end{proof}
  
  By iteratively applying the previous claim over items in a set $\hat{X} \subseteq T_i^t \setminus X$, we obtain
   \begin{claim}\label{cl:near-rational2}
  Let $X \subseteq S_i^t$ be such that $\alpha \cdot v_i(X) \geq p^t(X)$. 
  Then, for every $\hat{X} \subseteq T_i^t \setminus X$, we have
$\alpha \cdot v_i(X \cup \hat{X}) \geq p^t(X\cap S_i^t)+ p^{t+1}(X \cup \hat{X} \setminus S_i^t)$. \qed
\end{claim}
  
Now take any $\hat{S}\subseteq S_i^{t+1}$. To complete the proof of Theorem~\ref{thm:smra-truthful}, we must show that
$\alpha\cdot v_i(\hat{S}) \geq p^{t+1}(\hat{S})$. 
For this purpose, set $S'=\hat{S} \cap S_i^t$ and set $T'=\hat{S}\setminus S_i^t$. 
By the induction hypothesis, we have that $\alpha \cdot  v_i(S') \geq p^t(S')$.

Thus we may apply the second claim with $X=S'$ and $\hat{X}=T'$ to obtain
 \begin{eqnarray*} 
  \alpha \cdot v_i(\hat{S}) &=&
 \alpha \cdot v_i(S' \cup T')\\
 &\geq& p^t(S'\cap S_i^t)+ p^{t+1}((S' \cup T') \setminus S_i^t)\\
 &=&p^t(S')+ p^{t+1}(T')
   \end{eqnarray*} 

Furthermore, note that $S'\subseteq S_i^{t}\cap S_i^{t+1}$. In order to be
the provisional winner of an item~$j$ in both rounds $t$ and $t+1$, it must be the case that
no other bidder bid for item $j$ at the price $p^{t+1}(j)$. 
Thus the price of $j$ at time $t+1$ remains $p^{t}(j)$. Hence $p^{t+1}(S')=p^{t}(S')$, and so
 \begin{eqnarray*} 
  \alpha \cdot v_i(\hat{S}) 
  &\ge& p^t(S')+ p^{t+1}(T')\\
 &=& p^{t+1}(S')+ p^{t+1}(T')\\
 &=&p^{t+1}(\hat{S})
   \end{eqnarray*} 
Theorem~\ref{thm:smra-truthful} follows by induction.
\qed
\end{proof}

We remark that this proof implies a stronger conclusion: if bidder $i$ is truthful then
she is $\alpha$-individually rational {\em regardless} of the strategies of other bidders.

\subsection{A Welfare Guarantee}\label{sec:welfare-truthful}
\begin{theorem}\label{lem:ce}
Given $\alpha$-near-submodular truthful bidders, the SMRA outputs an allocation $\mathcal{S}=(S_1,\dots, S_n)$
with social welfare $\omega(\mathcal{S}) \geq \frac{1}{1+\alpha} \cdot \omega(\mathcal{S}_i^*)$
 where $S^*=(S_1^*,\ldots,S_n^*)$ is an allocation of maximum welfare.
\end{theorem}
\begin{proof}
Assume the auction terminates in round $t$ with a set of prices ${\bf p}^t$.
Thus $T_i^t=\emptyset$ for each bidder $i$. In particular, by truthfulness,
we have that
\begin{equation*}
v_i(S_i \cup (S_i^*\setminus S_i)) - p^t(S_i^* \setminus S_i)  \ \le\  v_i(S_i \cup \emptyset) - p^t(\emptyset)
\ =\  v_i(S_i)
\end{equation*}
 Thus
\begin{equation*}
 v_i(S_i^*) \ \leq\  v_i(S_i \cup S_i^*) 
 \ \leq \  v_i(S_i) + p^t(S_i^* \setminus S_i)
\end{equation*}
We now obtain a $(1+\alpha)$ factor welfare guarantee.
  \begin{eqnarray*} 
  \sum v_i(S_i^*) 
  &\leq& \sum_{i=1}^n \left( v_i(S_i) + p^t(S_i^* \setminus S_i)\right) \\
   &\leq& \sum_{i=1}^n  v_i(S_i) + \sum_{i=1}^n p^t(S_i^*) \\
   &\leq& \sum_{i=1}^n  v_i(S_i) + \sum_{i=1}^n p^t(S_i) \\
  &\leq&  \sum_{i=1}^n v_i(S_i) + \sum_{i=1}^n \alpha\cdot v_i(S_i)\\
  &=& (1+\alpha)\cdot \sum_{i=1}^n v_i(S_i) 
   \end{eqnarray*} 
Here the third inequality follows because the SMRA mechanism utilizes provisional winners. This implies that 
every item with a positive price is sold at the end of the auction. Consequently, $\sum_{i=1}^n p^t(S_i) \ge \sum_{i=1}^n p^t(S_i^*)$.
The fourth inequality follows as the auction allocation is $\alpha$-individual rational, as shown in Theorem \ref{thm:smra-truthful}. 
\qed
\end{proof}

By combining Theorem~\ref{thm:smra-truthful} and Theorem~\ref{lem:ce} we obtain Theorem~\ref{thm:constantfactor}.

\subsection{Tightness of the Bicriteria Guarantees}\label{sec:tight-truthful}

The bounds in Theorem \ref{thm:constantfactor} are almost tight.
To see this, consider the following example.
There are $k$ items $X=\{x_{1}, x_{2}, \dots, x_k\}$.
Let there be a large number $L$ of identical bidders.
For any $S\subset X$, each bidder $i$ has a valuation:
$$
v_i(S) = 
\begin{cases} 1 &\mbox{if } |S|=1 \\ 
(|S|-1)\cdot \alpha+ 1 & \mbox{if }  |S|\ge 2
\end{cases} 
$$
It is easy to verify that this function is $\alpha$ near-submodular.

The optimal welfare is obtained by allocating the entire set $X$ to a single bidder achieving social welfare $(k-1)\cdot \alpha+1$.
Now let us examine the allocation produced by the SMRA. Initially, all prices are $0$ and the truthful bid for each bidder is to demand 
the entire set $X$. 
Indeed, every bidder keeps bidding on the entire set (except for the items that she is the standing high bidder) until every item has price 
greater than $\frac{1}{k}\left((k-1)\cdot \alpha+ 1\right)$. At this point, no profitable bids
can be made and all bidders drop out.

In each round, the randomly chosen standing high bidders are all distinct
with probability at least $(1-\frac{k-1}{L-1})^{k}$. For $L>>k$, this probability tends to $1$. So by the end of the 
auction, the $k$ items are allocated to $k$ different bidders with probability almost $1$.
Since the social welfare of this allocation is only $k$, the expected social welfare of the SMRA is around $k$.
When $k$ goes to infinity, the welfare ratio tends to $\alpha$.

Next consider the rationality of this allocation. Each winner was allocated 
exactly one item with probability almost $1$, and the final price of that item is $\frac{1}{k}\left((k-1)\cdot \alpha+ 1\right)$.
The bidder has only value $1$ for the item. When $k$ goes to infinity, this tends to $\alpha$-rationality for the winners.
We remark that even for $k=2$ items, the previous example ensures that the welfare guarantee 
cannot be improved beyond $\frac{\alpha}{2}$ since the optimal welfare is $(\alpha+1)$ and the expected welfare
of the SMRA is $2$. 

\section{Bicriteria Guarantees under Locally Optimal Bidding}\label{sec:locallyopt}

As discussed in Section~\ref{sec:truthful-smra}, the assumption of truthful bidding
is unrealistic in the SMRA. Consequently, here we examine an alternate natural bidding method.
Given $S^{t-1}_i$, a bid $T_i^t\subseteq \Omega\setminus S^{t-1}_i$ 
is \emph{locally optimal} if $v_i(S^{t-1}_i\cup T_i^t)- p^t(T_i^t) \ge v_i(S^{t-1}_i\cup X)- p^t(X)$ for all 
$X\subseteq \Omega \setminus S_i^{t-1}$, where $|X\setminus T_i^t|\le 1$ and $|T_i^t\setminus X|\le 1$.
Observe that a locally optimal bid can be obtained via a local improvement algorithm 
that, given the current solution, seeks to add one item, delete one item, or replace one item.
Analysing this local improvement method is useful because local comparison
is a key tool used by bidders in real bandwidth auctions. Thus, there are practical reasons
to suspect that bidders will not make bids that are clearly not locally optimal.
From the theoretical viewpoint, this specific local improvement method is interesting because
it is guaranteed, given any set of prices, to output an optimal
set if the valuations satisfy the gross substitute property \cite{GS99}.

Now if we assume that bidders bid on locally optimal sets, we can still obtain bicriteria guarantees
on both the welfare and the rationality of the mechanism. 

\restatethm{\ref{thm:localimprov}}{
If bidders have $\alpha$-near-submodular valuations and make locally optimal bids,
then the SMRA has welfare ratio $\frac{1}{1+\alpha^2}$ and is $\alpha$-individually rational.}
 
\begin{proof}
We first argue that the allocation of the SMRA is $\alpha$-individually rational.

\begin{lemma}\label{lem:smra-local}
If bidders have $\alpha$-near-submodular valuations and make locally optimal bids,
then the SMRA has welfare ratio $\frac{1}{1+\alpha^2}$ and is $\alpha$-individually rational.
\end{lemma}
\begin{proof}
The proof is the same as that of Theorem \ref{thm:smra-truthful}. Truthfulness was used
to prove Inequality \ref{eq:useful} in Claim \ref{cl:near-rational1}. Observe, however,
that truthfulness is not necessary; locally optimality is sufficient to prove Inequality \ref{eq:useful}
since we just need that the utility of $T_i$ is better than the utility of any subset 
$T'$ of $T_i$ such that $|T_i \setminus T'| = 1$. Moreover no condition on valuation functions
is used in Claim~\ref{cl:near-rational2}. \qed
\end{proof}

Next, we show that the social welfare of the SMRA is at least $\frac{1}{1+\alpha^2}$ of the optimal welfare.

\begin{lemma}\label{lem:ce-local}
If bidders have $\alpha$-near-submodular valuations and make locally optimal bids
 the SMRA outputs an allocation $S=(S_1,\dots, S_n)$ with social welfare
 \begin{equation*} 
 \sum_{i=1}^n v_i(S_i) \geq \frac{1}{1+\alpha^2} \cdot \sum_{i=1}^n v_i(S_i^*) 
 \end{equation*}
 where $S^*=(S_1^*,\ldots,S_n^*)$ is an allocation of maximum welfare.
\end{lemma}
\begin{proof}
Assume the auction terminates in round $t$ with a set of prices ${\bf p}^t$.
Thus $T_i^t=\emptyset$ for each bidder $i$. 
 Let $S_i^*\setminus S_i =\{x_1,x_2,\dots, x_\ell\}$, say.
By local optimality,
we have, for any $x_j\in S_i^*\setminus S_i$, that
\begin{eqnarray*}
v_i(S_i \cup x_j) - p^t(x_j) &\le& v_i(S_i \cup \emptyset) - p^t(\emptyset)\\
&=& v_i(S_i)
\end{eqnarray*}
 Thus 
 \begin{equation} \label{eq:marg}
 v_i(S_i \cup x_j) - v_i(S_i) \le p^t(x_j)
 \end{equation}
 We then have
 \begin{eqnarray*} 
 v_i(S_i^*) &\leq&  v_i(S_i \cup S_i^*) \\
 &\leq&  v_i(S_i) + \sum_{j=1}^\ell  \left(v_i(S_i\cup \{x_1,\dots, x_j\}) - v_i(S_i\cup \{x_1,\dots, x_{j-1}\}) \right) \\
 &\leq&  v_i(S_i) + \alpha\cdot \sum_{j=1}^\ell  \left(v_i(S_i\cup x_j) - v_i(S_i) \right) \\
 &\leq&  v_i(S_i) +\alpha\cdot \sum_{j=1}^\ell p^t(x_j) \\
 &=& v_i(S_i) +\alpha\cdot p^t(S_i^*\setminus S_i) 
 \end{eqnarray*}
 Here the third inequality follows by $\alpha$ near-submodularity. The fourth inequality comes from
 Inequality~(\ref{eq:marg}).
We finally obtain a $(1+\alpha^2)$ factor welfare guarantee:
  \begin{eqnarray*} 
  \sum v_i(S_i^*) 
  &\leq& \sum_{i=1}^n \left( v_i(S_i) + \alpha\cdot p^t(S_i^* \setminus S_i)\right) \\
   &\leq& \sum_{i=1}^n  v_i(S_i) + \alpha\cdot \sum_{i=1}^n p^t(S_i^*) \\
   &\leq& \sum_{i=1}^n  v_i(S_i) + \alpha\cdot \sum_{i=1}^n p^t(S_i) \\
  &\leq&  \sum_{i=1}^n v_i(S_i) + \alpha \cdot\sum_{i=1}^n \alpha\cdot v_i(S_i)\\
  &=& (1+\alpha^2)\cdot \sum_{i=1}^n v_i(S_i) 
   \end{eqnarray*} 
Here the third inequality follows because the SMRA mechanism utilizes provisional winners. This implies that 
every item with a positive price is sold at the end of the auction. Consequently, $\sum_{i=1}^n p^t(S_i) \ge \sum_{i=1}^n p^t(S_i^*)$.
The fourth inequality follows as the auction allocation is $\alpha$-individual rational, as shown in Lemma~\ref{lem:smra-local}. 
\qed
\end{proof}
By combining Lemmas~\ref{lem:smra-local} and~\ref{lem:ce-local}, we obtain our theorem.
\end{proof}

\subsection{Tightness of Bicriteria Guarantees}

The bounds in Theorem \ref{thm:localimprov} are essentially tight.
To see this, consider the following example.
There are $k \cdot n+1$ items. Specifically, there is a special item $z$ and,
for each $i \in [n]$, there is a collection of $k$ items $X_i=\{x_{i,1}, x_{i,2}, \dots, x_{i,k}\}$.

There will be two classes of bidders.
First, there are $n$ Type I bidders. 
Bidder $i\in [n]$ only values the items $X_i=\{x_{i,1}, x_{i,2}, \dots, x_{i,k}\}$. 
For any $S\subseteq X_i$, she has a valuation:
$$
v_i(S) = 
\begin{cases} \alpha &\mbox{if } |S|=1 \\ 
(|S|-1)\cdot \alpha^2+\alpha & \mbox{if } |S|\ge 2  
\end{cases} 
$$
Her marginal value is always zero for any item not in $X_i$. This function is 
$\alpha$ near-submodular.

There are $k \cdot n \cdot L$ Type II bidders, where $L>>k,n$.
For each $i\in [n]$ and each $j\in [k]$, there are $L$ identical bidders that only value
the set $\{x_{i,j}, z\}$.
Each such bidder $\ell$ has a valuation
function with $v_\ell(x_{i,j})=1, v_\ell(z)=H$ and $v_\ell(\{x_{i,j},z\})=H+\alpha$,
where $H$ is an integer larger than $\alpha$ that we will specify later. 
Moreover, her marginal value is always zero for any other item. 
Again, these valuation functions are $\alpha$ near-submodular.

Together, we have $(k \cdot L+1)\cdot n$ bidders. 
The optimal welfare is obtained by allocating each set $X_i$ to the Type I bidder $i$ and
$z$ to any Type II bidder. This allocation has social welfare $\left((k-1)\cdot \alpha^2+\alpha)\right) \cdot n + H$.

Now consider the allocation produced by the SMRA. Initially at ${\bf p=0}$, the unique
locally optimal bid is for each bidder to bid every item in their demand set. Thus a Type I bidder 
demands $X_i$ and a Type II bidder demands $\{x_{i,j}, z\}$.
This bidding behavior will remain until every item has price greater than $\alpha$
(Type II bidders still bid since $H > \alpha$). Let us call this round time $t$.

After time $t$, the locally optimal bid for each Type I bidder $i$ is to demand the empty set.
To see this, note that $p^t(x_{i,j}) > v_i(x_{i,j}) =\alpha$ for each item in $x_{i,j} \in X_i$. 
Each Type $I$ bidder drops out. 

On the other hand, since $L>>k,n$, we may assume that the randomly chosen standing high  
bidder for {\em every} item  in each round is Type II (this happens with probability almost $1$). In particular, at time $t$, the standing high bidder for each item in $X_i$ 
is Type II. After time $t$, Type II bidders only bid on item $z$ until its price reaches $H$.
As a result, by the end of the SMRA, every item is allocated to some Type II bidder.
The total welfare of the SMRA auction is then at most $(kn-1)\cdot 1 + (H+\alpha)$.
For $n>> H$ and sufficiently large $k$ this gives a welfare ratio that tends to $\alpha^2$.

Next consider the rationality of this allocation. Amongst the ``winners", each Type II bidder (except at most one)
wins exactly one item. The final price of that item in some $X_i$ is $\alpha$ but the bidder has a valuation
for the item of one. Thus, all these bidders are only $\alpha$-individually rational.

\section{Individually Rational Bidding}\label{sec:ir}

So, as shown in Theorems~\ref{thm:smra-truthful} and~\ref{thm:localimprov}, truthful and 
locally-optimal bidding can only ensure approximate individual rationality in the SMRA. 
Consequently, such bidding strategies are highly risky. In this section, we investigate
what bidding strategies are risk-free and what are the welfare implications of such strategies.

We call a risk-free strategy {\em conservative}, and show in Section~\ref{sec:conservativesecure}
that conservative bidding is (weakly) characterized by {\em secure bidding}. 
Specifically, secure bidding always produces individually rational outcomes.
Conversely, if the other bidders use secure bids then the only way a bidder can ensure
an individually rational solution is by also bidding securely. This result holds even with stronger
assumptions on the bidding strategies of the other bidders, for example, that they make profit-maximizing
secure bids.

We then examine the welfare consequences of secure bidding. Our main result, in
Section~\ref{sec:welfare-secure}, is that 
then the welfare ratio is at most $1+\alpha$ provided the bidders
make profit maximizing secure bids.
This result is surprising in that we are able to match the welfare guarantee of truthful bidding
without having to lose individual rationality.


\subsection{Secure Bidding}\label{sec:conservativesecure}
We say that a bidding strategy is {\em conservative} if it cannot lead to a bidder
having negative utility. Thus, conservative strategies are individually rational.
To understand what strategies are conservative, we first need to understand
what constitutes a bidding strategy. In the SMRA, a bidder can select a bid based upon the
auction history she observed, for example, the sequence of price vectors, her sequence of conditional bids,
and on her sequence of provisional sets of items. Thus, we consider a bidding strategy to be a function 
of these three factors.\footnote{In some SMRA mechanisms, bidders also know the excess demand of each item.}

We say that a conditional bid $T_i^t$ is \emph{secure} for bidder $i$ (given the provisional 
winning set $S_i^{t}$) if $v_i(S') \geq  p(S')$ for every $S'\subseteq S_i^{t} \cup T_i^t$. 
A bidding strategy is {\em secure} if every conditional bid it makes is secure. It is easy to verify that 
any secure bidding strategy is individually rational. 
We now show that  bidding securely in every round is essentially the only individually rational strategy.


\begin{lemma}\label{lem:conservativegeneral}
Let $t$ be an integer and $T_i^{\hat{t}},S_i^{\hat{t}}$, ${\bf p}^{\hat{t}}$ be 
the conditional bid of bidder $i$, the provisional winning set of bidder $i$ and the price 
vector at round $\hat{t}$ for any $\hat{t} \leq t$. If bidder $i$ makes a non-secure bid in round $t+1$, 
then there exist secure bidders who can bid consistent with the history and 
ensure that bidder $i$ has negative utility in the final allocation.
\end{lemma}
\begin{proof}
Assume that the conditional bid $T_i^{t+1}$ of bidder $i$ at some round $t$ is not secure, 
then there exists $S' \subseteq S_i^{t+1} \cup T_i^{t+1}$ such that $S'$ 
satisfies $v_i(S') < p_i^{t+1}(S')$. Let us prove that there exists an 
auction such that, with high probability, (i) the set allocated to $i$ is $S'$, (ii) at 
any time $\hat{t} \leq t$, the provisional winning set of $i$ is $S_i^{\hat{t}}$ and 
(iii) the price vector at round $\hat{t}$ is ${\bf p}^{\hat t}$.

The auction is as follows: there are many copies of the same bidder $1$ whose 
valuation function is $v_1$. 
Let $M$ be an integer larger than the maximum of the prices at any round $\hat{t} \leq t$ 
and the maximum valuation of any subset of items for bidder $i$. The valuation 
function $v_1$ of all the copies of bidder $1$ is additive\footnote{A valuation 
function $v$ is additive if $v(S)= \sum_{s \in S} v(s)$.} and the value of each item is the following:
$$
v_1(s) = 
\begin{cases} M + 2 \cdot \epsilon &\mbox{if } s \in \Omega \setminus S'\\ 
p^t(s) & \mbox{if } s \in S' \setminus S_i^t \\
p^t(s)-\epsilon & \mbox{if } s \in S' \cap S_i^t
\end{cases} 
$$

\begin{claim}
Assume that $i$ bids on $T_i^{\hat{t}}$ at any round $\hat{t} \leq t$. There is a sequence of 
secure bids such that, for every $\hat{t} \leq t$, with high probability
\begin{enumerate}[(i)]
 \item the price vector is exactly ${\bf p}^{\hat{t}}$ at the end of round $\hat{t}$, 
 \item bidder $i$ is the standing high bidder of the set $S_i^{\hat{t}+1}$.\footnote{Recall that $i$ is the standing high bidder of $S_i^{\hat{t}}$ at the beginning of round $\hat{t}$, which explains the index difference.} 
\end{enumerate}
\end{claim}
\begin{proof}
By induction on $t$, let us prove that if the copies of bidder $1$ use the following strategy, 
the conclusion holds. 
If the price of item $s$ does not increase from round $\hat{t}$ to round $\hat{t}+1$
then no copy of bidder $1$ bids on it at round $\hat{t}$;
if the price of item $s$ increases and $s \in S_i^{\hat{t}}$ then no copy of bidder $1$ bids 
on it at round $\hat{t}$;
if the price of $s$ increases and $s \notin S_i^{\hat{t}}$ then all the copies of bidder $1$ 
bid on it at round $\hat{t}$. 

By construction of the valuation function $v_1$, at any time $\hat{t} \leq t$, the value of any 
item $s \in \Omega \setminus (S' \cap S_i^{t+1})$ for copies of bidder $1$ is at least its price. 
Moreover, if $s \in S_i^t$ then copies of bidders $1$ do not bid on $s$ at price $p^t(s)$ by 
construction. It is easy to verify that $v_1(s)$ is larger than the price of $s$ at any 
round where copies of $1$ bid on it.
As $v_1(\cdot)$ is additive, all bids by copies of bidder $1$ up to round $t$ are secure.

Let us show that items in excess demand are those whose prices increase between 
rounds $\hat{t}-1$ and $\hat{t}$. If the price of an item in $\Omega \setminus S_i^{\hat{t}}$ is 
distinct in ${\bf p}^{\hat{t}-1}$ and ${\bf p}^{\hat{t}}$, then all the copies of bidder $1$ bid on it, 
and it is in excess demand. Now assume $s \in S_i^{\hat{t}}$, if the price of $s$ increases, 
then $s \in S_i^{\hat{t}} \setminus S_i^{\hat{t}-1}$ (the provisional winner must change when 
there is a price increment). Thus bidder $i$ bids on $s$ and then $s$ is in excess demand.

Now let us show that with high probability, bidder $i$ is the standing high bidder of the 
items in $S_i^{\hat{t}}$. Since the prices of any item $s$ in $S_i^{\hat{t}-1} \cap S_i^{\hat{t}}$ 
do not increase, copies of bidder $1$ do not bid on $s$ at round $\hat{t}$. Thus 
$s$ is still in $S_i^{\hat{t}}$. 
Moreover, bidder $i$ is the unique bidder in excess demand for the items 
in $S_i^{\hat{t}} \setminus S_i^{\hat{t}-1}$. So the provisional set of bidder $i$ contains $S_i^{\hat{t}}$. 
Let us prove that it does not contain any other item $s$ with high probability. First assume 
that $s \in S_i^{\hat{t}-1} \setminus S_i^{\hat{t}}$. Thus the price of $s$ increases. And since $i$ 
was the standing high bidders of these items at round $\hat{t}-1$, she cannot be the standing 
high bidder anymore at round $\hat{t}$. Assume now that $i$ bids on 
$s \notin S_i^{\hat{t}-1} \cup S_i^{\hat{t}}$. Then by construction, all the other copies 
of $1$ also bid on $s$ and then, with high probability (since there are many copies of 
bidder $1$), $s$ is not allocated to bidder $i$, which completes the proof of the claim. 
\qed
\end{proof}

Now assume that at round $t+1$, bidder $i$ decides to bid on $T_i^{t+1}$. Starting 
from round $t+1$, copies of bidder $1$ securely bid on subsets in the 
complement of $S'$ until the prices of all items in $\Omega\setminus S$ reach $M+2\epsilon$. 
Note that since no copy of $1$ bid on any item in $S'$, all the items in $S'$ are in the provisional 
set of $i$ at the end of round $t+1$.
Copies of $1$ continue to perform the same bids until they drop out. On the other hand, bidder $i$ 
can perform any bid.

Let us first show that the set allocated to $i$ contains $S'$. At the end of round $t+1$, the 
price of item $s$ in $S'$ is $p^t(s)$ if $s \in S' \cap S_i^t$ and $p^t(s)+\epsilon$ if $s \in S' \setminus S_i^t$. 
Thus the price of $s$ is above $v_1(s)$ and then copies of $1$ cannot bid anymore on $s$ since they 
make secure bids. Since $S' \subseteq S_i^{t+1}$, the set of items allocated to $i$ by the SMRA 
contains the set $S'$.

Assume now that $s \notin S'$ is allocated to $i$ at the end of the procedure. Since copies of $1$ 
continue to bid on it until its price is at least $M+ \epsilon$. This implies that bidder $i$ bids on it 
at price at least $M+\epsilon$. Thus the price of the set allocated to $i$ is at least $M+\epsilon$, 
which is above the value of any set for bidder $i$ by definition of $M$. So $i$ is not individually 
rational. Otherwise, bidder $i$ is allocated the set $S'$ and by definition of $S'$, we have 
$p^t(S') >v_i(S')$ and then bidder $i$ receives negative utility.\qed
\end{proof}

So, if the bids of the other bidders are secure, then performing a non-secure bid may lead to negative 
utility. One may ask if a similar statement still holds if  stronger assumptions 
are made concerning the strategies of the other bidders. 
This is indeed the case. The following lemma states that even if we know the other
bidders are truthful (or if they make profit-maximizing secure bids),
making any non-secure bid is not individually rational.
Bidder $i$ performs a \emph{profit-maximizing secure bid} $T_i$ if the bid is secure 
and the utility of $S_i \cup T_i$ is maximized over all possible secure bids.

\begin{lemma}\label{lem:conservative truthful}
 Let $t$ be an integer and $T_i^{\hat{t}}, S_i^{\hat{t}}$, ${\bf p}^{\hat{t}}$ be the 
 conditional bid of bidder $i$, the provisional winning set of bidder $i$ and the price vector 
 at round $\hat{t}$ for any $\hat{t} \leq t$. Assume that there is an item of 
 value $0$ for $i$ with price $\epsilon \cdot \hat{t}$ at any round $\hat{t} \leq t$.
If bidder $i$ makes a non-secure bid in round $t+1$, 
then there exist truthful (or profit-maximizing secure) bidders who can bid consistent 
with the history and ensure that bidder $i$ has negative utility in the final allocation.
 \end{lemma}

\begin{proof}
Let us construct an auction such that at any round $\hat{t} \leq t$, the price vector is ${\bf p}^{\hat{t}}$ and the provisional winning set of bidder $i$ is $S_i^{\hat{t}}$.
Assume that at some round $t$, the bid of $i$ is not secure. Then there exists $S' \subseteq S_i^t \cup T_i^{t+1}$ that satisfies $v_i(S') < p(S')$.

\paragraph{Instance of the SMRA.}
Let us now construct an auction such that $S'$ is allocated to $i$. Let us denote by $s$ the item of value $0$ for bidder $i$ such that 
$p^{\hat{t}}(s)= \epsilon \cdot \hat{t}$ for any $\hat{t} \leq t$. Before describing formally the instance, let us give some intution.
There are two main types of bidders.
First we create bidders for time periods $\hat{t} \leq t$. For any item $s'$ whose price increases at round $\hat{t}$ 
and such that $i$ does not bid on $s'$ at round $\hat{t}$, we create unit-demand bidders that bid on $s$ in the first $\hat{t}-1$ rounds 
and bid  on $s'$ at round $\hat{t}$. These bidders ensure that the price vector is ${\bf p}^{\hat{t}}$ at any round $\hat{t}$ smaller than $t$. 
Second, we create bidder for time period $\hat{t}>t$. These bidders will ensure that that the set allocated to $i$ is $S'$. Indeed, they will bid
on items in the complement of $S'$ until we are sure that, if $i$ still bids on them, the strategy of $i$ is not conservative. 
The most technical part of the proof consists in constructing the first type of bidders.

Let $s'$ be an item distinct from $s$ such that $p^{\hat{t}}(s') \neq p^{\hat{t}-1}(s')$. Assume moreover that $s' \notin S_i^{\hat{t}}$. Then we create three copies of a bidder $b$ such that:
$$
v_b(S) = 
\begin{cases} 
p^{\hat{t}}(s') &\mbox{if } S=\{s'\}\\ 
\epsilon \cdot \hat{t} & \mbox{if } S = \{ s \}
\end{cases} 
$$
Moreover, we assume that if the utility of both items is the same, bidder $b$ prefers item $s'$. This preference rule can also be simulated by making a small modification to the valuation function. However, we present it using the preferences rule, as we believe it makes the proof cleaner. These three bidders are called the \emph{initial bidders of round $\hat{t}$ for item $s'$}.
The initial bidders are the union of the initial bidders of round $\hat{t}$ for $\hat{t} \leq t$. The initial bidders will permit to fit the price vector at any round $\hat{t} \leq t$.

Let $M$ be the maximum of the value of a subset of items for bidder $i$ and of $\epsilon \cdot t$.
We can now create the second type of bidder that will push the prices of items in $\Omega \setminus S'$ after time $t$. For each item $s' \in \Omega \setminus S'$, we create three copies of the same bidder such that the value of $s'$ is $2 (M+\epsilon)$ and the value of $s$ is $M+\epsilon$. Such bidders are called the \emph{final bidders}.  The final bidders will permit to ensure that the set $S'$ will be the set allocated to bidder $i$ at the end of the auction. \vspace{5pt}

Let us prove the following simple facts:
\begin{claim}
Let $\hat{t}\leq t$. Assume that at round $\hat{t}-1$ of the auction, the price vector is ${\bf p}^{\hat{t}-1}$ then: 
 \begin{itemize}
  \item All the initial bidders of round $\hat{t}' < \hat{t}$ have an empty conditional bid at round $\hat{t}$.
  \item If the initial bidders of round $\hat{t}' > \hat{t}$ for item $s'$ bid on $s'$ at round $\hat{t}$ then the price of $s'$ must increase at any step between $\hat{t}$ and $\hat{t}'$.
  \item All the final bidders bid on $s$ at round $\hat{t}$.
  \item If $s'$ is an item such that $p^{\hat{t}}(s')-p^{\hat{t}-1}(s')=\epsilon$ and $s' \notin S_i^{\hat{t}}$ then there are initial bidders of round $\hat{t}$ bidding on $s'$.
 \end{itemize}
\end{claim}
\begin{proof}
Let $b$ be an initial bidder of round $\hat{t}'$ for item $s'$  with $\hat{t}' < \hat{t}$. The price of $s$ at round $\hat{t}$ is $\epsilon \hat{t}$ which is larger than the value of $s$ for $b$ which is $\epsilon \cdot \hat{t}'$. Moreover, since we have created initial bidder of round $\hat{t}'$ for item $s'$, the price of $s'$ increases between rounds $\hat{t}'$ and $\hat{t}'+1$. And by definition initial bidders, the value of $s'$ for $b$ is $p^{\hat{t}'}(s')$. Thus the price of $s'$ at time $\hat{t}$ is larger than the value of $s'$ for $b$. Since $b$ make secure bids, she does not bid on any item and then has an empty conditional bid at round $\hat{t}$.

Let $\hat{t}' > \hat{t}$ and $b$ be an initial bidder of round $\hat{t}'$ for item $s'$. At round $\hat{t}$, the utility of bidder $b$ for item $s$ is $\epsilon \cdot \hat{t}' - \epsilon \cdot \hat{t}= \epsilon \cdot (\hat{t}'-\hat{t})$. And the utility of $b$ for $s'$ is $p^{\hat{t}'}(s')-p^{\hat{t}}(s')$. This difference is $\epsilon$ times the number of rounds between $\hat{t}$ and $\hat{t}'$ where the price of $s'$ increases. Thus it at most $\epsilon \cdot (\hat{t}'-\hat{t})$ and we have equality if and only if the price of $s'$ increases at any round between $\hat{t}$ and $\hat{t}'$. If the equality does not hold, then the utility of $s$ is larger than the one of $s'$. Otherwise, the preference rule ensures that $s'$ is prefered, which proves the second point.

The proof of the third point is straightforward. Assume that $b$ is a final bidder for item $s'$. By definition of $M$, the utility of $b$ on $s$ is at least $M + 2\epsilon$ since $p^t(s) \leq M$. On the other hand, the utility of $s'$ is at most $M+\epsilon$. Thus $b$ bids on $s'$.

Let us finally prove the last point. According to the definition of the instance, we have created three initial bidders of round $\hat{t}$ for the item $s'$. Given ${\bf p}^{\hat{t}-1}$, these bidders have utility $\epsilon$ for both $s$ and $s'$ at round $\hat{t}$. By definition of the preference rule, if these bidders are allocated the empty set, they prefer bidding in $s'$ rather than in $s$. Since they bid on at most two items $s$ and $s'$, one of these three bidders has an empty provisional set and then bid on $s'$ at round $\hat{t}$.
\end{proof}

\paragraph{Running the auction.}
Let us prove by induction that, with positive probability, at any round $\hat{t} \leq t$, the price vector is ${\bf p}^{\hat{t}}$ and the provisional set of bidder $i$ is $S_i^{\hat{t}}$. For $t=0$ the statement immediately holds. Now assume that the price vector of the auction fits ${\bf p}^{\hat{t}-1}$ after round $\hat{t}-1$ and that $i$ is the standing high bidder of the items in $S_i^{\hat{t}-1}$. For each item $s'$ such that the price of $s'$ increases between $\hat{t}-1$ and $\hat{t}$ and such that $s' \notin  S_i^{\hat{t}}$, we have created initial bidders for $s'$ of round $\hat{t}$. Then the last point of the claim ensures that  the price of $s'$ increases between round $\hat{t}-1$ and $\hat{t}$ increases. 
Moreover, since these items are in excess demand and $i$ does not bid on them, bidder $i$ is not the standing high bidders on  these items are round $\hat{t}$.

Now consider any item $s'$ in $S_i^{\hat{t}-1}$ that is still in $S_i^{\hat{t}}$. No initial bidder of round $\hat{t}$ was created for item $s'$. Moreover, all the initial bidders of round $\hat{t}' > \hat{t}$ for item $s'$ do not bid for $s'$ by the second point of the claim. Indeed the price of $s'$ does not increase at any step between $\hat{t}$ and $\hat{t}'$. 
So there is no excess demand on $s'$ and then the price of $s'$ and its standing high bidder remain the same. A similar argument ensures that no item $x$ such that $p^{\hat{t}'-1}(x)= p^{\hat{t}'}(x)$ is in excess demand. Thus the price of $x$ is not modified.

Let us finally consider the items $s'$ in $S_i^{\hat{t}} \cap T_i^{\hat{t}-1}$. Since the condition bid of $i$ contains $s'$, $s'$ is in excess demand. Moreover, since $i$ is a candidate to be the new provisional winner of such an item, $s' \notin S_i^{\hat{t}-1}$. However other bidders may also be candidates to be provisional winner of $s'$ (for instance initial bidders for $s'$ of later rounds). Since the provisional winner is chosen uniformly at random amongst the candidates, bidder $i$ is chosen with positive probability. 

So, for any round $\hat{t} \leq t$, there is a positive probability that at any round $\hat{t} \leq t$ the price vector is ${\bf p}^{\hat{t}}$ and the provisional winning set of $i$ is $S_i^{\hat{t}}$. Now at round $t+1$, bidder $i$ bids on $T_i^{t+1}$. The behavior of the auction after this round is different (and actually simpler).
Indeed, the claim ensures that no initial bidder makes any conditional bid anymore. Moreover
\begin{itemize}
 \item  Since the prices of items in $S'$ are larger than their values for any other bidders, no bidder will bid on any item in $S'$ after round $t$. Thus the set of items allocated to $i$ contains the set $S'$.
 \item Since there exist a lot of bidders whose value on each item in $\Omega \setminus S'$ is larger than the value of any set for bidder $i$, these bidders will continue to bid on these items until bidder $i$ drops out. Thus no item of $\Omega \setminus S'$ will be allocated to $i$ at the end of the procedure.
\end{itemize}
It completes the proof of Lemma~\ref{lem:conservative truthful}. \qed
\end{proof}

\subsection{Social Welfare under Secure Bidding}\label{sec:welfare-secure}

The previous results ensure that secure bidding strategies are essentially the only way to guarantee $1$-individual rationality. In this section, we will assume that bidders strategies are secure.
The following simple lemma ensures that any allocation where each bidder is allocated at least one 
item can be obtained in the SMRA with bidders only making secure bids.

\begin{lemma}\label{lem:allocnonempty}
 Any allocation where each bidder is allocated at least one item can be obtained via the SMRA with secure bidders. In particular if each bidder
 is allocated at least one item then the optimal allocation can be obtained if bidders are secure.
\end{lemma}
\begin{proof}
Let $\mathcal{S}=(S_1,\ldots,S_n)$ be an allocation of the items where $S_i\neq\emptyset$ is allocated to bidder $i$.
 Then this allocation can be obtained with secure bidders. Indeed, assume 
 that at round $t=0$ each bidder simply bids on the set $S_i$. 
 Since all the items have price $0$, all the bids are trivially secure. Then, at the end of first step,
 every bidder $i$ is the standing high bidder of the set $S_i$ and no item is in excess demand. Thus the SMRA 
 stops and allocates to each bidder $i$ the set $S_i$.
 \qed
\end{proof}

Lemma~\ref{lem:allocnonempty} is unsatisfactory in two ways. First, if there are bidders 
that are allocated nothing then the situation can be far more complex.
Specifically, it may then be the case, see Lemma~\ref{lem:noguaranteesuperadditive}, that 
secure bidding cannot provide a guarantee on welfare.

\begin{lemma}\label{lem:noguaranteesuperadditive}
If super-additive bidders only make secure bids, then there is no guarantee on the social welfare of the SMRA.
\end{lemma}
\begin{proof}
Let $M$ be a positive integer. Consider  an auction with three bidders $\{0,1,2\}$ 
and two items $\{a,b\}$. The bidders $\{0,1\}$ are unit-demand.
The value of item $a$ for bidder $0$ is equal to $1$ and the value of item $b$ for 
bidder $1$ is equal to $1$. Bidder $2$ has a super-additive valuation function.
Its value for each individual item is $\frac 12$ but it has value $M$ for the set $\{ a,b\}$.
Observe that if bidder $2$ bids securely, then she cannot bid on the items $a$ or $b$ once their 
prices rise beyond $\frac12$. Therefore the welfare is only $1$ with item $a$ allocated to bidder $0$ and item $b$ allocated to bidder $1$.
Clearly, the optimal allocation has welfare $M$ which can be
arbitrarily large.
\qed
\end{proof}

The second unsatisfactory aspect of Lemma~\ref{lem:allocnonempty} is that
the structure of the bids used there is extremely artificial, since 
the bidders need to know all the valuation functions in order to calculate the secure bids. 
Theorem~\ref{thm:profitmax} shows we can circumvent both of these
problems if the bidders' valuation functions are $\alpha$-near-submodular.
Then a good welfare guarantee can be obtained if the bidders make 
profit-maximizing secure bids.

\restatethm{\ref{thm:profitmax}}{
Given bidder valuation functions that are $\alpha$ near-submodular. 
Assume moreover that all the set values are multiple of $\epsilon$. If each bidder
bids for a profit-maximizing (conditional) secure set in every round then the SMRA outputs
a solution $\mathcal{S}$ with welfare $\omega(\mathcal{S})\ge \frac{1}{1+\alpha}\cdot  \omega(\mathcal{S}^*)$.
}

\notshow{\begin{proof}
Let $\mathcal{S}^*=\{S^*_1,\dots, S^*_n\}$ be the optimal allocation, and let $\mathcal{S}=\{S_1,\dots, S_n\}$
be the assignment output by the SMRA. By assumption, $S_i$ is the most profitable secure set in the final round,
and the conditional bid $T_i$ is empty in the final round. Thus $S_i$ was the provisional set for bidder $i$ 
in the penultimate round.

We claim that $v_i(S_i)+ p(S^*_i)\ge \frac{1}{4\alpha}\cdot v_i(S^*_i)$.
We may assume that $v_i(S_i)\le \frac12 v_i(S^*_i)$. If not, the claim is trivially satisfied.

Let $S^*_i \setminus S_i =\{x_1, x_2, \dots, x_k\}$ and let $X^j=\{x_1, x_2, \dots, x_j\}$, for each $j\le k$.
Then, for any $W\subseteq S_i$, we have
\begin{eqnarray*}
\frac12\cdot v_i(S^*_i) &\le& v_i(S_i\cup X^k)- v_i(S_i) \\
&=& \sum_{j=1}^k  \left( v_i(S_i\cup X^j)-v_i(S_i\cup X^{j-1})\right) \\
&\le& \sum_{j=1}^k  \alpha \cdot \left( v_i(W\cup x_j)-v_i(W)\right)
\end{eqnarray*}
Here the first inequality follows by the free-disposal property.
The second inequality  is due to $\alpha$- near-submodularity.
Now define $m_j=v_i(S_i\cup X^j)-v_i(S_i\cup X^{j-1})$. Thus,
$\sum_{j=1}^k m_j = v_i(S_i \cup S_i^*) - v_i(S_i) \ge \frac12\cdot v_i(S^*_i)$  since $v_i(S_i)\le \frac12 v_i(S^*_i)$.

Let $Y$ be the set of indices of items in $S^*_i \setminus S_i$ for which
$p(x_j) < \frac{1}{2\alpha}\cdot m_j$. First suppose that $Y$ is empty.
Then
\begin{eqnarray*}
p(S^*_i) &\ge& \sum_{j} p(x_j) \\
&\ge& \sum_{j} \frac{1}{2\alpha}\cdot  m_j \\
&\ge& \frac{1}{4\alpha}\cdot v_i(S^*_i)
\end{eqnarray*}
But then $v_i(S_i)+ p(S^*_i)\ge \frac{1}{8\alpha}\cdot v_i(S^*_i)$, as desired.

Hence, we may assume that $Y$ is non-empty. So take any $j\in Y$ and
consider the set $W\cup x_j$, for any $W\subseteq S_i$.
By assumption, we have $m_j> 2\alpha\cdot p(x_j)$. And $\alpha\cdot
\left( v_i(W\cup x_j)-v_i(W)\right) \ge m_j$ by $\alpha$ near-submodularity. 
Together this implies that $\left( v_i(W\cup x_j)-v_i(W)\right)> 2 p(x_j)$.
Thus we may bound the profitability of the set $W\cup x_j$ as follows:
\begin{eqnarray}
v_i(W\cup x_j)-p(W\cup x_j)
&=& v_i(W)+\left(v_i(W\cup x_j)-v_i(W) \right) -p(W\cup x_j)\nonumber \\
&>&  v_i(W)+2p(x_j) -p(W\cup x_j) \nonumber \\
&=& v_i(W)-p(W) +p(x_j)  \nonumber\\
&\ge& v_i(W)-p(W) \label{eq:more-profit}\\
&\ge& 0 \label{eq:secure}
\end{eqnarray}
Two observations follow. First plugging in $W=S_i$ into inequality (\ref{eq:more-profit})
shows that the set $S_i\cup x_j$ is more profitable than the set $S_i$.
Second, inequality (\ref{eq:secure}) shows that every subset of $S_i\cup x_j$
is profitable. Thus the bid $S_i\cup x_j$ is secure and more profitable than
$S_i$, a contradiction.

The claim follows. Therefore, we have that
$\sum_i \left(v_i(S_i)+ p(S^*_i)\right) \ge \frac{1}{4\alpha}\cdot \sum_i v_i(S^*_i)$.
Note that $\omega(\mathcal{S})=\sum_i v_i(S_i)$.
Moreover, recall, that in the SMRA every item with a positive price is sold. Thus, by the disjointness of
$\{S^*_1,\dots, S^*_n\}$, we have that $\omega(\mathcal{S})\ge \sum_i p(S^*_i)$.
Consequently, we have that $\omega(\mathcal{S})\ge \frac{1}{8\alpha}\cdot \omega(\mathcal{S}^*)$.
\qed
\end{proof} }

\begin{proof}
Let $\mathcal{S}^*=\{S^*_1,\dots, S^*_n\}$ be the optimal allocation, and let $\mathcal{S}=\{S_1,\dots, S_n\}$
be the assignment output by the SMRA. By assumption, $S_i$ is the most profitable secure set in the final round,
and the conditional bid $T_i$ is empty in the final round and then $S_i$ was the provisional set for bidder $i$ 
in the penultimate round.
Let $S^*_i \setminus S_i =\{x_1, x_2, \dots, x_k\}$ and let $X^j=\{x_1, x_2, \dots, x_j\}$, for each $j\le k$.

For every item $x_{j}\in S_{i}$, since the conditional bid $T_{i}$ is empty, there are two possibilities.
\begin{itemize}
\item {\bf Case 1:} $\{x_{j}\}$ is a secure conditional set  but not as profitable as $\emptyset$. Then $v_{i}\big(S_{i}\cup \{x_{j}\}\big)-p(x_{j}) -  \epsilon < v_{i}(S_{i})$. Let $Q_{ij}$ be $S_{i}$, and we have $p(x_{j}) \geq v_{i}\big(Q_{ij}\cup\{x_{j}\}\big)-v_{i}(Q_{ij})$ since values are multiple of $\epsilon$.
\item {\bf Case 2:} $\{x_{j}\}$ is an insecure conditional set. Then there exist a set $Q\subseteq S_{i}$ such that $v_{i}\big(Q\cup\{x_{j}\}\big) < p\big(Q\cup\{x_{j}\}\big)$. On the other hand, since $S_{i}$ is a secure set, $v_{i}\big(Q\big) \geq p\big(Q\big)$. Let $Q_{ij}$ be $Q$, and we have $p(x_{j})\geq v_{i}\big(Q_{ij}\cup\{x_{j}\}\big)-v_{i}(Q_{ij})$.
\end{itemize}
In both case, we have $p(x_{j})\geq v_{i}\big(Q_{ij}\cup\{x_{j}\}\big)-v_{i}(Q_{ij})$. Using these inequalities, we can bound $v_{i}(S_{i}^{*})$.
\begin{eqnarray*}
v_{i}(S_{i}^{*})&\le& v_i(S_i\cup X^k)\\
&=& v_{i}(S_{i})+\sum_{j=1}^k  \left( v_i(S_i\cup X^j)-v_i(S_i\cup X^{j-1})\right) \\
&\le& v_{i}(S_{i})+\sum_{j=1}^k  \alpha \cdot \left(v_{i}\big(Q_{ij}\cup\{x_{j}\}\big)-v_{i}(Q_{ij})\right)\\
&\le& v_{i}(S_{i})+\alpha\cdot\sum_{j=1}^k p(x_{j})\\
&\le& v_{i}(S_{i})+\alpha\cdot p(S_{i}^{*})\end{eqnarray*}
The second inequality is because $Q_{ij}\subseteq S_{i}$ and $v_{i}(\cdot)$ is $\alpha$-near-submodular. The third inequality is derived from the case analysis above.

Finally, we are ready to bound the welfare ratio.
  \begin{eqnarray*} 
  \sum v_i(S_i^*) 
   &\leq& \sum_{i=1}^n  v_i(S_i) + \alpha\cdot \sum_{i=1}^n p(S_i^*) \\
   &\leq& \sum_{i=1}^n  v_i(S_i) + \alpha\cdot \sum_{i=1}^n p(S_i) \\
  &\le& (1+\alpha)\cdot \sum_{i=1}^n v_i(S_i) 
   \end{eqnarray*} 
The last inequality holds because under secure bidding, the allocation is individually rational. \qed
\end{proof}

The bound in Theorem~\ref{thm:profitmax} is almost tight. 
This can be seen by adapting the example in Section~\ref{sec:tight-truthful}.

Thus, under secure bidding we are able to match the welfare guarantee of truthful bidding
without having to lose individual rationality. This suggests that secure bidding might be 
the best strategy to use in an SMRA auction. Unfortunately, this is probably not the case.
Recall, from Section \ref{sec:truthful-smra}, that in addition to the basic ascending price
mechanism the SMRA has an associated set of bidding activity rules to encourage truthful bidding.
As discussed, the rules are actually too weak to ensure truthful bidding. But the rules
are strong enough to make secure bidding very risky.
In particular, each bidder has an amount of eligibility points. The larger the number of points the
bigger the collection of items a bidder may bid on. A bidder loses eligibility points if she
bids on a small set -- such bids will then hurt the bidder in later rounds. 
Observe that secure bidding naturally favours bidding upon small sets and is, thus, risky.

\end{document}